\documentclass[12pt]{iopart}

\usepackage{graphicx}
\expandafter\let\csname equation*\endcsname\relax
\expandafter\let\csname endequation*\endcsname\relax
\usepackage{amsmath}
\usepackage{amsthm}  
\usepackage{amsfonts}
\usepackage{amssymb} 
\newtheorem{thm}{Theorem}[section]

\newtheorem{definition}[thm]{Definition}

 \usepackage{hyperref}
  \usepackage{xcolor}
\begin{document}

\title[New multi-hump exact solitons of a KdV system with conformable derivative]{New multi-hump exact solitons of a coupled Korteweg-de-Vries system with conformable derivative describing shallow water waves via RCAM}

\author{Prakash Kumar Das  
}

\address{Department of Mathematics, Trivenidevi Bhalotia College,\\ Raniganj, Paschim Bardhaman,
West Bengal, India-713347}
\ead{prakashdas.das1@gmail.com}
\vspace{10pt}
\begin{indented}
\item[]August 2020
\end{indented}

\begin{abstract}
In this article, a modification of the rapidly convergent approximation method is proposed to solve a  coupled Korteweg-de Vries equations with conformable derivative that govern shallow-water waves. Based on the Leibniz and chain rule of conformable derivative, these equations reduced into ODEs with integer-order using traveling wave transformation. Adopting the modified scheme  a new novel exact solution of the reduced coupled ordinary differential equations is obtained in terms of exponential functions. Finally, by putting them back into traveling wave transformation the solutions of the considered  partial differential equations with conformable derivative are derived. To ensure the boundedness of the derived solutions few theorems have been proposed and proved. The derived results of the theorems are utilized to plot the solutions. Graphics exhibit that solutions have variant multi-hump soliton peculiarities and their tails decay to zero exponentially in a monotonic manner. These results not only show the efficiency of the modified scheme but also establish that the solution is enriched with new multi-hump features. 
\end{abstract}

%

\section{Introduction}
\label{intro}
One-pulse solutions  are  common and significant features of dispersive partial differential equations recounting physical systems.  Beyond that, there may exist two-pulse, three-pulse, and generally n-pulse (multi-hump) solutions for such physical systems having a higher-order dispersion ( involving  higher-order differential operators). Moreover, when a spatial soliton is built of multiple modes of the initiated waveguide, its  soliton intensity profile may be compounded and display several peaks.  The multi-pulse solutions constitute with multiple copies of the one-pulse solutions separated by finitely many oscillations close to the zero equilibrium. These solutions have been derived numerically for various nonlinear models as discussed below.

Numerically it is showed that there may exist infinitely many multi-pulse solutions  to the fifth order Korteweg -de Vries (KdV) equation \cite{buffoni1996global,groves1998solitary}. Also, the existence of an infinite (but countable) set of localized stationary solutions,  experimentally and numerically established for nonlinear generalised KdV equation in \cite{gorshkov1979existence}. For solitary waves governed by incoherent beam interaction in a saturable medium Ostrovskaya et al., \cite{ostrovskaya1999stability}  disclosed that two-hump solitary waves are linearly stable in a wide region of their existence, but all three-hump solitons are linearly unstable. For nonintegrable multi-component nonlinear models formation of stable multi-hump solitary waves, a novel physical mechanism  has  been presented in \cite{ostrovskaya2001multi}. These multi-hump optical solitons observed numerically  in the models describing laser radiation copropagating with a Bose-Einstein condensate and used to shed light on the phenomenon of jet emission from a condensate interacting with a laser \cite{cattani2011multihump}.

Besides numerical methods, few direct methods also exist in literature to investigate multi-hump solitons. In \cite{parra2019multi} the Hirota bilinear method has been employed  to find one- and two-hump exact bright and dark soliton solutions to a coupled system between the linear Schr$\ddot{\text{o}}$dinger and KdV equations. Using  the Lax pair and Darboux transformation (DT), the authors in \cite{wang2016breather} have derived  multi-peak soliton.  Most of the multi-pick solitons studied in literature are using numerical schemes. So it is of great importance to propose a new method which can derive exact multi-hump soliton solutions efficiently and a simple way than existing direct methods. To attain the goal here for the first time the rapidly convergent approximation method (RCAM) has been employed to construct exact multi-hump solutions of a coupled differential equations with conformable derivative.

Non-linear frational differential equations (NLFDE) model abundant branches of applied mathematics describing evolution of physical processes in science and engineering \cite{magin2006fractional,engheta1996fractional,schneider1989fractional,chen2005fractional,meral2010fractional,he1997semi,he2004variational,banerjee2017study,das2018time,kilbas2006theory,kumar2019new,kumar2020analysis,prakash2019analysis,prakash2018new,goyal2020regarding}. The availability of solutions of such equations contribute  a lot in visualize the fractional system involved. In lots of cases, it is arduous task to derive exact solution to these equations. Despite that reserchers have proposed a few direct methods \cite{aslan2019optical,inc2018dark,yusuf2020optical,khalil2014new,pandir2018analytical,ellahi2018exact,sabi2019new,houwe2020solitary,rezazadeh2019hyperbolic,osman2019traveling} to look for  exact travelling wave solutions of NLFDE. 
The above mention direct methods often make few assumptions of the form of the solution, as a result, they always unable to produce a profound type solution. In addition to that, these methods methodology always require to solve a system of nonlinear algebraic equations, which is a hard task and becomes improbable when nonlinearity increases. Whenever these systems of equations are solvable, they lead to some particular solutions. As a consequence, these methods yield a set of special solutions instead of a general one. Also,  there are many approximation methods \cite{hashemi2020three,korpinar2020theory,korpinar2020residual,korpinar2018numerical,yusuf2019invariant,yong2008numerical,atangana2013time,biswas2019approximate,osman2019new,arqub2020numerical,kumar2019hybrid,goswami2019efficient,prakash2019numerical} to deal with equations for which the direct methods do not work.  These approximate methods are often found to be slowly convergent and unable to provide the close form of the series solution. These problems can be easily tackled by the RCAM  \cite{das2015improved,das2016rapidly,das2018rapidly,das2018solutions,das2018piecewise,das2019some,das2019rapidly,das2020chirped}. In this article, this scheme is used to obtain a new multi-hump travelling wave solution of a coupled Korteweg-de Vries equations  with conformable derivative.

The famous couple KdV equation is a nonlinear frequency dispersion equation, which describes shallow long-wave and small-amplitude phenomena, ion acoustic waves in a plasma, acoustic waves on a crystal lattice and so on. For the time being its generalisation containing  time-fractional derivative getting great deal attention to the researchers whose form is:
\begin{align}\label{KDVe0}
 & D_t^{\alpha}u(x,t)-p\ u_{xxx}(x,t)-p_1 \ u(x,t)\ u_x(x,t) -p_2 \ v(x,t)\ v_x(x,t)=f(x,t), \nonumber \\
 & D_t^{\alpha}v(x,t)+q \ v_{xxx}(x,t)+q_1 \ u(x,t)\ v_x(x,t)=g(x,t),\ 0< \alpha \leq 1,
\end{align}
where  $D_{*}^{\alpha}$ denote the Caputo type differential operator and $f(x,t)$ and $g(x,t)$ are two source terms respectively.
Eq. (\ref{KDVe0}) play  an important role in the propagation of waves and have significant contribution in many applied fields such as fluid, mechanics, plasmas, crystal lattice vibrations at low temperatures, etc. The Eq. (\ref{KDVe0}) without source term ($f(x,t)=g(x,t)=0$) and $p=a,\ p_1 =6a,\ p_2=2b, \ q=-1,\ q_1=3$ was first proposed and obtained its approximate series solution by using Adomian decomposition method  in \cite{yong2008numerical}. Latter this equation was considered by many authors \cite{jin2010new,atangana2013time,ray2013new,bulut2015numerical} and they have applied the generalised differential transform method, homotopy decomposition method, fractional reduced differential transform method and Haar wavelets respectively to obtain the approximate solutions. Matinfar et al. \cite{matinfar2015functional} applied  the functional variable method for deriving the exact solution of  the said equations.
Three numerical technique based on the shifted Legendre polynomials, Meshless spectral method and spectral collection method were applied in \cite{bhrawy2016numerical,hussain2019meshless,albuohimad2018numerical} for obtaining the  numerical solution of the Eq. (\ref{KDVe0}) with source terms.
In a recent work  S. Biswas et al., \cite{biswas2019approximate} proposed another variant of coupled fractional KdV equation containing both space-time fractional derivative. They derived it with the help of a semi-inversion method, variational principle,  and  Lagrangian of the KdV equation. In addition, the authors employed the homotopy analysis method to derive the approximate solution of the equation. The said equation enjoy the form
\begin{align}\label{KDVe00}
 & D_t^{\alpha}u(x,t)+6 \ a \ u(x,t)\ D_x^{\alpha}u(x,t)-2\ b\ v(x,t)\ D_x^{\alpha}v(x,t)+a\ D_x^{3\alpha}u(x,t)=0, \nonumber \\
 & D_t^{\alpha}v(x,t)+3\ u(x,t)\ D_x^{\alpha}v(x,t)+D_x^{3\alpha}v(x,t)=0,\ \ 2<3\alpha<3.
\end{align}
Equations  (\ref{KDVe0}) and  (\ref{KDVe00}) are defined using  Caputo  derivative, which  do not satisfy some main principles of classical integer order derivative  \cite{tarasov2016chain,tarasov2013no} such as Leibniz rule, chain rule and etc. So it is not straightforward to derive exact solutions to these equations involving this derivative. Hence in this article, we consider the equation (\ref{KDVe00})  with conformable  derivative in the form
 \begin{align}\label{KDVe1}
 & u^{(\beta)}_t(x,t)+6 \ a \ u(x,t)\ u^{(\alpha)}_x(x,t)-2\ b\ v(x,t)\ v^{(\alpha)}_x(x,t)+a\ u^{(\alpha)}_{xxx}(x,t)=0, \nonumber \\
 & v^{(\beta)}_t(x,t)+3\ u(x,t)\ v^{(\alpha)}_x(x,t)+v^{(\alpha)}_{xxx}(x,t)=0,\ \ 0<\alpha, \  \beta \leq 1.
\end{align}
Since the system (\ref{KDVe1}) contains conformal derivatives is equivalent to the classical system of integer-order KdV-type equations with variable coefficients of a special kind \cite{tarasov2018no}. This  coupled KdV equations with variable coefficients represents a simple generalization of Hirota-Satsuma coupled KdV equations  \cite{hirota1981soliton, zhou2003periodic}, describe the interaction of two long waves \cite{zhou2003periodic, xie2004exact,singh2006lie} and also has many applications in the above mentioned fields. Also in the case $\alpha=\beta=1$, the above equations reduces to conventional Hirota-Satsuma coupled KdV equations.   In this assignment, the author's intention is to construct the exact multi-hump solutions of Eq.(\ref{KDVe1}) by utilizing RCAM.


Section \ref{sec2} presents the basic properties of  conformable derivative. Basic methodologies of RCAM are introduced in section \ref{sec3}. Using this scheme a class of new travelling wave solutions for coupled KdV equation  with conformable derivative were presented in section  \ref{sec4}. The boundedness conditions of solutions of the KdV equation  with conformable derivative and its reduced ordinary differential equations  have been presented in section \ref{sec5} and verified through plotting them. The outlook of the present work have been summarized and some concluding remarks given in section \ref{sec6}.		
\section{Properties of conformable derivative}\label{sec2}
This section presents few definitions and features of conformable derivative \cite{khalil2014new,abdeljawad2015conformable}:
\begin{definition}The conformable derivative of a function $f : [0,\infty ) \rightarrow \mathbb{R}$ of order $\alpha$ is defined by
\begin{eqnarray}\label{FDe1}
T_{\alpha}(h)(t)= \lim_{\epsilon \rightarrow 0} \frac{h\left( t+\epsilon\ t^{1-\alpha} \right)-h(t)}{\epsilon},
\end{eqnarray}
for all $t>0, \ \alpha \in (0,1).$ If f is $\alpha$-differentiable in some $(0, a), a > 0,$ and $ \lim_{t \rightarrow 0+} f^{(\alpha)}(t)$ exists, then define
 $f^{(\alpha)}(0)=\lim_{t \rightarrow 0+} f^{(\alpha)}(t).$
\end{definition}
 Sometimes, we use the notation $f^{(\alpha)}(t)$ in place of $T_{\alpha}(f)(t)$, to prevail the conformable derivatives of $f$ of order $\alpha$. Few significant features of conformable  derivative are presented below:\\
If  $\alpha \in (0,1],$ and $f, g$ be $\alpha$-differentiable at a point $t>0$ then we have
\begin{itemize}
\item[1.]   $T_{\alpha} \left( a \ f+b\ g\right)=a\ T_{\alpha}(f)+b\ T_{\alpha}(g),$ for all $a,b \in \mathbb{R}.$ \\
\item[2.]  $T_{\alpha}\left( t^p \ \right)=p\ t^{p-\alpha}$ for all $p \in \mathbb{R}.$\\
\item[3.]  $T_{\alpha}(\lambda)=0$, for all constant functions $f(t)=\lambda.$\\
\item[4.]  $T_{\alpha}\left(f g \right)=f\ T_{\alpha}(g)+g\ T_{\alpha}(f).$ \\
\item[5.]  $T_{\alpha}\left(\frac{f}{g}\right)=\frac{g\ T_{\alpha}(f)+f\ T_{\alpha}(g)}{g^2}$.\\
\item[6.] If $f$ is differentiable, then  $T_{\alpha}(f)(t)=t^{1-\alpha} \frac{df}{dt}(t).$
\end{itemize}
\begin{thm}[Chain Rule \cite{abdeljawad2015conformable,eslami2016exact,chen2018simplest}] Let $f,g : (0 , \infty)\rightarrow \mathbb{R}$ be two  differentiable  functions and also $f$ is $\alpha$-differentiable, then, one has the following rule:\\
 $T_{\alpha}( fog ) ( t ) = t^{ 1-\alpha} \ g'( t ) f'( g ( t ) ).$
 \end{thm}
\section{The modified RCAM} \label{sec3}
Consider a system of equations in the form
\begin{eqnarray}\label{eq2p1}
 {\cal X}^{'''}-{\cal A}^{2}.\; {\cal X}^{'}={\cal N},
\end{eqnarray}
where   ${\cal X }, {\cal A }$ and ${\cal N}$  are matrix of   dependent variables, constant coefficients and nonlinear terms  respectively, given by 
 $${\cal X}=\left[\begin{array}{c}U_1(x)\\U_2(x) \\ \vdots \\U_k(x) \\\end{array}\right] , \ \ {\cal A}=\left[\begin{array}{c}
 \lambda_1\ \ 0\ \ \cdots 0 \\ 0\ \ \lambda_2 \ \ \cdots 0 \\ \vdots \\ 0\ \ 0\ \ \cdots \ \ \lambda_k  \\\end{array}\right], \ \text{and}  \ {\cal N}=\left[\begin{array}{c}
 N_1\left(U_1(x),\cdots,U_k(x)\right) \\ N_2\left(U_1(x),\cdots,U_k(x)\right) \\ \vdots \\N_k\left(U_1(x),\cdots,U_k(x)\right) \\\end{array}\right].$$ 
Note that here considered all  $\lambda_i, \ i=1,2, \cdots, k$ are distinct. Which indicate the modification over the existing scheme  \cite{das2019rapidly}.  To find the solution of (\ref{eq2p1}), we recast it  in exponential matrix operator form
\begin{equation}\label{eq2p2}
\hat{{\cal O}}[{\cal X}](x) = {\cal N},
\end{equation}
where linear exponential matrix operator can be recast in the form
\begin{equation}\label{eq2p3}
\hat{{\cal O}}[\cdot](x)  = e^{ {\cal A}.x}\frac{d}{dx}\left(e^{- 2 {\cal A}.x}\frac{d}{dx}\left(e^{ {\cal A}.x}\frac{d}{dx}[\cdot]\right)\right).
\end{equation}
The inverse operator $\hat{{\cal O}}^{-1}$ of the operator $\hat{{\cal O}}[](x)$ is  given by
\begin{eqnarray}\label{eq2p4}
\hat{{\cal O}}^{-1}[\cdot](x)  = \int^{x} e^{- {\cal A}.x^{'}}\int^{x^{'}} e^{2 {\cal A}.x^{''}}\int^{x^{''}} e^{- {\cal A}.x^{'''}}[\cdot] dx^{'''}  dx^{''} dx^{'}.
\end{eqnarray}
Operating $\hat{{\cal O}}^{-1}$ on $\hat{{\cal O}}[ {\cal X}](x)$ and using integration by parts produces
\begin{eqnarray}\label{eq2p5}
\hat{{\cal O}}^{-1}\left( {\cal X}^{'''}- {\cal A}^{2}.\;  {\cal X}^{'}\right)  = {\cal X}- {\cal C}. e^{ {\cal A}. x}- {\cal D}.e^{- {\cal A}. x}-{\cal E},
\end{eqnarray}
where ${\cal C}=\left[\begin{array}{c}c_1 \\c_2 \\ \vdots \\c_k \\\end{array}\right]$, ${\cal D}=\left[\begin{array}{c}
 d_1 \\  d_2 \\ \vdots  \\d_k \\\end{array}\right]$ and ${\cal E}=\left[\begin{array}{c}
 e_1 \\ e_2 \\ \vdots  \\e_k \\\end{array}\right]$ are integration constant matrices.
Operating $\hat{{\cal O}}^{-1}$  both sides of (\ref{eq2p2}) and using (\ref{eq2p5}), gives
\begin{equation}\label{eq2p6}
{\cal X} ={\cal E}+{\cal C}.e^{{\cal A}. x} +{\cal D}.e^{-{\cal A}. x}+\hat{\cal O}^{-1}[{\cal N}](x),
\end{equation}
involving three arbitrary constants matrices ${\cal C}$, ${\cal D}$ and ${\cal E}$. To derive the terms involving the unknown ${\cal X}$ in R.H.S of (\ref{eq2p6}), we express them in the form
\begin{eqnarray}\label{eq2p7}
{\cal X} & \cong &\left[\begin{array}{c}U_1(x) \\ U_2(x) \\ \vdots  \\U_k(x) \\\end{array}\right]= \sum_{m=0}^\infty \left[\begin{array}{c}U_{1,m}(x) \\ U_{2,m}(x) \\ \vdots  \\U_{k,m}(x) \\\end{array}\right],
\end{eqnarray}
and
${\cal N}=\sum_{m=0}^\infty \Delta_m (x),$
with
\begin{eqnarray}\label{eq2p8}
 \Delta_m \cong \left[\begin{array}{c}\Delta_{1,m}(x) \\ \Delta_{2,m}(x) \\ \vdots  \\\Delta_{k,m}(x) \\ \end{array}\right] = \frac{1}{m!}\left[\frac{d^m}{d\epsilon^m}\left[\begin{array}{c}N_{1}(x) \\N_{2}(x) \\ \vdots \\N_{k}(x) \\\end{array}\right]\right]_{\epsilon=0},
\end{eqnarray}
and $N_{j}(x), \ j=1,2, \cdots,k$ are given by  $$N_{j}(x)=\left(\sum_{k=0}^\infty U_{1,k}\epsilon^k,\sum_{k=0}^\infty U_{2,k}\epsilon^k, \cdots  ,\sum_{k=0}^\infty U_{k,k}\epsilon^k \right).$$ The terms $\Delta_{i,m} (x)= \Delta_{i,m} (U_{1,0}(x),U_{1,1}(x), …..,U_{1,m}(x), \cdots, U_{k,0}(x),U_{k,1}(x)$ $,…..,U_{k,m}(x)),\ i=1,2,\cdots ,k$ are  Adomian polynomials \cite{adomian1994solving,duan2011new,adomian1983inversion,adomian1993analytic,adomian1993new,adomian1994modified} turned out from the formula (\ref{eq2p8}).
Use of (\ref{eq2p8}) in (\ref{eq2p9}) gives
\begin{eqnarray}\label{eq2p11}
&&{\cal X} =  \left[ \begin{array}{c} c_{1}\; e^{\lambda_1 x}+ d_{1} \; e^{-\lambda_1 x}+e_{1} \\ c_{2}\; e^{\lambda_2 x}+ d_{2} \; e^{-\lambda_2 x}+e_{2} \\ \vdots \\ c_{k}\; e^{\lambda_k x}+ d_{k} \; e^{-\lambda_k x}+ e_{k} \\\end{array}\right] + \hat{\cal O}^{-1}\left[\sum_{m=0}^\infty  \left[\begin{array}{c}\Delta_{1,m}(x) \\ \Delta_{2,m}(x) \\ \vdots  \\\Delta_{k,m}(x) \\ \end{array}\right] \right]. \nonumber \\
\end{eqnarray}
We follow the steps given in \cite{das2019rapidly}, to obtain the higher order correction term as
\begin{eqnarray}\label{eq2p9}
&{\cal X}_0 \cong \left[\begin{array}{c}U_{1,0}(x)\\U_{2,0}(x)\\ \vdots \\U_{k,0}(x) \\\end{array}\right]= \left[\begin{array}{c} c_{1}\; e^{\lambda_1 x}+ d_{1} \; e^{-\lambda_1 x}+e_{1} \\ c_{2}\; e^{\lambda_2 x}+ d_{2} \; e^{-\lambda_2 x}+e_{2} \\ \vdots \\ c_{k}\; e^{\lambda_k x}+ d_{k} \; e^{-\lambda_k x}+ e_{k} \\\end{array}\right],
\end{eqnarray}
\begin{eqnarray}\label{eq2p10}
&  {\cal X}_{n+1} \cong \left[\begin{array}{c}U_{1,n+1}(x)\\ U_{2,n+1}(x)\\  \vdots \\U_{k,n+1}(x) \\\end{array}\right]= \hat{\cal O}^{-1}\left[\begin{array}{c}\Delta_{1,n}(x) \\ \Delta_{2,n}(x) \\ \vdots  \\\Delta_{k,n}(x) \\ \end{array}\right],\ n \geq 0.
\end{eqnarray}
In the following, we present a few special cases of iterative formulas (\ref{eq2p9})-(\ref{eq2p10}).\\  
\textbf{Case-I} In case $ \lambda_i>0,\ i=1,\cdots,m,$ and $\lambda_i<0,\ i=m+1,\cdots,k $ for $0<m<k $, use of the vanishing boundary condition $U_{i}(-\infty) = 0 $  in (\ref{eq2p11}) for the localized solution impart $d_{i} =  e_{i} = 0,\ \text{for} \ i=1,2, \cdots, m $ and $c_{i} =  e_{i} = 0,\ \text{for} \ i=m+1,m+2, \cdots, k $. Consequently the leading and higher order rectification terms  of the series solution are given by (\ref{eq2p10}) with
\begin{eqnarray}\label{eq2p11}
&{\cal X}_0 \cong \left[\begin{array}{c}U_{1,0}(x) \\ \vdots \\U_{m,0}(x)\\U_{m+1,0}(x) \\ \vdots \\U_{k,0}(x) \\\end{array}\right]= \left[\begin{array}{c} c_{1}\; e^{\lambda_1 x} \\ \vdots \\ c_{m}\; e^{\lambda_m x}\\ d_{m+1}\; e^{-\lambda_{m+1} x}\\ \vdots \\ d_{k}\; e^{-\lambda_k x} \\\end{array}\right].
\end{eqnarray}
\textbf{Case-II} In case all $ \lambda_{i} > 0 $, use of the  condition $U_{i}(-\infty) = 0 $  in (\ref{eq2p11}) for the localized solution confers $d_{i} =  e_{i} = 0,\ i=1,2, \cdots, k $. Hence the leading and higher order amendment terms  of the series solution are given by (\ref{eq2p10}) with
\begin{eqnarray}\label{eq2p14}
&{\cal X}_0 \cong \left[\begin{array}{c}U_{1,0}(x)\\U_{2,0}(x)\\ \vdots \\U_{k,0}(x) \\\end{array}\right]=\left[\begin{array}{c} c_{1}\; e^{\lambda_1 x} \\ c_{2}\; e^{\lambda_2 x}\\ \vdots \\ c_{k}\; e^{\lambda_k x} \\\end{array}\right].
\end{eqnarray}
\textbf{Case-III} Likewise in case when all $ \lambda_{i} < 0 $, for boundary condition $U_{i}(-\infty) = 0 $ to obtain the localized solution  we go along with restraining the term involving $e^{\lambda_i x}$. In this case, the commanding and subsequent  terms  of the solution are given by (\ref{eq2p10}) with
\begin{eqnarray}\label{eq2p15}
&{\cal X}_0 \cong \left[\begin{array}{c}U_{1,0}(x)\\U_{2,0}(x)\\ \vdots \\ U_{k,0}(x) \\\end{array}\right]= \left[\begin{array}{c}  d_{1} \; e^{-\lambda_1 x} \\  d_{2} \; e^{-\lambda_2 x} \\ \vdots \\  d_{k} \; e^{-\lambda_k x} \\\end{array}\right].
\end{eqnarray}
Similarly for boundary condition $U_{i}(\infty) = 0 $, one needs to proceed contrary to the above cases.
Using above presented iteration schemes and symbolic software, one can obtain the general term of the series  (or generating function). Gives the exact solution of the considered system of equations.  
\section{Solution of  coupled KdV Eq. (\ref{KDVe1})}\label{sec4}
Let us consider the solution of the equation (\ref{KDVe1}) by using travelling wave transformation  \cite{inc2018dark,eslami2016exact,chen2018simplest} in the form
\begin{eqnarray}\label{KdVe2}
\begin{cases}
u(x,t)=U(\xi),\\
v(x,t)=V(\xi),
\end{cases}
\xi=\frac{k}{\alpha}\ x^{\alpha}+\frac{c}{\beta}\ t^{\beta}+\xi_0,
\end{eqnarray}
where $c,\ k$ and $\xi_0$ are constants. Using the chain rule of   conformable
derivatives, the transformation (\ref{KdVe2}) permit us to convert
(\ref{KDVe1}) to an ordinary differential equation in the
form :
\begin{eqnarray}\label{KdVe3}
\begin{cases}
U^{'''}(\xi)-\lambda_1^2\ U^{'}(\xi)+\frac{6 }{k^2}U(\xi) U^{'}(\xi)-\frac{2 b }{a k^2}V(\xi) V^{'}(\xi)=0,\\
V^{'''}(\xi)-\lambda_2^2\ V^{'}(\xi)+\frac{3}{k^2} U(\xi)  V^{'}(\xi)=0,
\end{cases}
\end{eqnarray}
where 
\begin{eqnarray}\label{KdVe4}
\lambda_1=\sqrt{-\frac{c}{a k^3}} \ \text{ and} \ \lambda_2=\sqrt{-\frac{c}{ k^3}}.
\end{eqnarray}
It is important to note here that throughout our discussion, we assume that  $\lambda_1$ and $\lambda_2$ are positive real and that leads us to the conditions  $a>0$ and $c k^3<0$. The above presented system of equations can be recast in the form (\ref{eq2p1}) with 
 $${\cal X}=\left[\begin{array}{c}U(\xi)\\V(\xi)  \\\end{array}\right] , \ \ {\cal A}=\left[\begin{array}{c}
 \lambda_1\ \ 0 \\ 0\ \ \lambda_2  \\\end{array}\right], \ \text{and} \ \ {\cal N}=\left[\begin{array}{c}
- \frac{6 }{k^2}U(\xi) U'(\xi)+\frac{2 b }{a k^2}V(\xi) V'(\xi) \\ - \frac{3}{k^2} U(\xi)  V'(\xi) \\\end{array}\right].$$ 
Thereafter, we pursue the imitating steps of RCAM presented in the section \ref{sec3} to construct the solution of (\ref{KdVe3}). To constitute  localized solutions satisfying boundary condition $U(-\infty)=V(-\infty)=0$,    we use (\ref{eq2p10}) with (\ref{eq2p14}).
 That yield the  following correction terms
 \begin{align*}
  U_0(\xi )=&c_1 \ e^{\lambda_1  \xi },\\
  V_0(\xi )=&c_2 \ e^{\lambda_2  \xi },\\
 U_1(\xi )=&- \left[ e^{2 \lambda _2 \xi } \left(a c_1^2 \left(\lambda _1^2-4 \lambda _2^2\right) e^{2 \left(\lambda _1-\lambda _2\right) \xi }+b c_2^2 \lambda _1^2\right)\right] \text{/}\left[ a k^2
   \left(\lambda _1^4-4 \lambda _1^2 \lambda _2^2\right)\right],\\
V_1(\xi )=&-\left[3 c_1 c_2 \lambda _2 e^{\left(\lambda _1+\lambda _2\right) \xi }\right] \text{/}\left[k^2 \lambda _1 \left(\lambda _1+\lambda _2\right)  \left(\lambda _1+2 \lambda _2\right)\right],\\
  U_2(\xi )=& \left[3 c_1 e^{\lambda_1 \xi } \left(a c_1^2 \lambda _2 \left(\lambda _1+\lambda _2\right)^2 \left(\lambda _1^2-4 \lambda _2^2\right) e^{2 \lambda _1 \xi }  +2 b c_2^2  \left( \lambda_1^2+2 \lambda _2^2\right) \lambda_1^3 e^{2 \lambda_2 \xi } \right) \right] \\
  & \text{/} \left[4 a k^4 \lambda_1^4 \lambda_2  \left( \lambda_1+\lambda_2 \right)^2 \left( \lambda_1^2-4 \lambda_2^2 \right) \right],\\
  V_2(\xi )=&\left[ c_2 e^{\lambda _2 \xi } \left(12 a c_1^2 \left(\lambda _1-2 \lambda _2\right) \lambda _2^3 e^{2 \lambda _1 \xi }  +b c_2^2 \left(\lambda _1+\lambda _2\right) \lambda _1^3
   e^{2 \lambda _2 \xi }\right)\right] \text{/} \left[8 a k^4 \lambda _1^3 \lambda _2^2 \right. \\
 &  \left. \times  \left(\lambda _1+\lambda _2\right)  \left(\lambda _1^2-4 \lambda _2^2\right)\right],\\
  U_3(\xi )=&-\left[2 a^2 c_1^4 \left(\lambda _1-2 \lambda _2\right)^2 \lambda _2^2 \left(\lambda _1+2 \lambda _2\right)^3 e^{4 \lambda _1 \xi }+12 a b  c_1^2 c_2^2 \lambda _2
   \left(\lambda _1^3-2 \lambda _2 \lambda _1^2 \right. \right. \\
   &  \left. +2 \lambda _2^2 \lambda _1-4 \lambda _2^3\right) \lambda _1^3 e^{2 \left(\lambda _1+\lambda _2\right) \xi }+ \left. b^2 c_2^4
   \left(\lambda _1+2 \lambda _2\right) \lambda _1^6 e^{4 \lambda _2 \xi }\right] \\
   &\text{/}\left[4 a^2 k^6 \lambda _1^6 \left(\lambda _1-2 \lambda _2\right)^2  \lambda _2^2 \left(\lambda _1+2 \lambda _2\right)^3\right],\\
  V_3(\xi )=&-\left[3 c_1 c_2 e^{\left(\lambda _1+\lambda _2\right) \xi } \left(2 a c_1^2 \lambda _2^2 \left(\lambda _1+2 \lambda _2\right)^2 \left(\lambda _1^2-\lambda _1 \lambda _2 - 2 \lambda _2^2\right) e^{2 \lambda _1 \xi } \right. \right. \\
  & \left. \left. +3 b c_2^2 \left(\lambda _1^2+\lambda _2 \lambda _1+2 \lambda _2^2\right) \lambda _1^4 e^{2 \lambda _2 \xi }\right)\right]\text{/} \left[8 a k^6 \lambda_1^5 \left(\lambda _1-2 \lambda _2\right) \lambda _2 \left(\lambda _1+\lambda _2\right)^2 \right. \\
  &\times \left.  \left(\lambda _1+2 \lambda _2\right)^3\right],
 \\
 \ \ \ \  \vdots 
 \end{align*}
where $c_1$ and $c_2$ are integration constants. Likewise, the  higher order terms can be calculated using symbolic computations available in software packages Mathematica and Maple. Afterward calculating up to eight or higher order iterations, the said software packages can easily provide the generating functions of those correction terms. For  this problem, we have obtained the following generating functions
 \begin{align}\label{mgfun1}
 U(\xi, \epsilon )=&4 k^4 \left(\lambda _1-2 \lambda _2\right) \left(\lambda _1+\lambda _2\right)^2 \left(\lambda _1+2 \lambda _2\right)^2 \left[ 64 c_1  a^2 k^8  \lambda_1^4 \lambda_2^4 \left(\lambda_1 -2  \lambda _2\right) \right. \nonumber  \\
 & \left. \times   \left(\lambda _1+\lambda _2\right)^2  \left(\lambda _1+2 \lambda _2\right)^4  e^{\lambda _1 \xi } - 16   \epsilon a b k^2 c_2^2 \lambda_2^2 \left(\lambda _1+2 \lambda _2\right)  \left\{ \epsilon c_1 k^2 \lambda_1^2  \right. \right.  \nonumber \\ 
 & \times  \left. \left. \left(\lambda_1^2-4 \lambda_2^2\right)^2  \left(\lambda_1^2+\lambda _2^2\right)   e^{\left(\lambda _1+2 \lambda _2\right) \xi } +\epsilon^2 c_1^2 \lambda _2^2 \left(\lambda _1^2-3 \lambda _1 \lambda _2  + 2   \lambda _2^2\right)^2  \right. \right.  \nonumber \\ 
 & \times  \left. \left.  e^{2 \left(\lambda_1+\lambda _2\right) \xi } +4 k^4 \lambda_1^4 \lambda_2^2 \left(\lambda _1^2+3 \lambda _2 \lambda _1+2 \lambda _2^2\right)^2  e^{2 \lambda _2 \xi }\right\}+ \epsilon^4 b^2  c_1 c_2^4 \lambda_1^4  \right. \nonumber \\
&\times \left. \left(\lambda _1-2 \lambda _2\right) \left(\lambda _1-\lambda _2\right)^2  e^{\left(\lambda _1+4 \lambda_2\right) \xi }\right]\text{/} Q(\xi, \epsilon)^2, 
\end{align}
\begin{align}\label{mgfun2}
V(\xi , \epsilon)=& 8 a c_2 k^4  \lambda_2^2 \left(\lambda _1-2 \lambda _2\right) \left(\lambda _1+\lambda _2\right) \left(\lambda _1+2 \lambda _2\right)^2   e^{\lambda _2 \xi } \left[  \epsilon  c_1 \left(\lambda _1^2-3 \lambda _2 \lambda _1 \right. \right. \nonumber \\
& \left. \left. +2 \lambda_2^2\right) e^{\lambda _1 \xi } +2 k^2 \lambda _1^2  \left(\lambda _1^2+3 \lambda _2 \lambda _1+2 \lambda _2^2\right) \right] \text{/} Q(\xi, \epsilon), 
\end{align}
   where
\begin{align}\label{mgfun3}
  Q(\xi , \epsilon)=& 8 a k^4 \lambda_2^2 \left(\lambda _1-2 \lambda _2\right)  \left(\lambda _1+\lambda _2\right)^2 \left(\lambda_1+2 \lambda_2\right)^3  \left( \epsilon c_1   e^{\lambda_1 \xi }+2 k^2 \lambda_1^2 \right)-b \epsilon^2 c_2^2 e^{2 \lambda _2 \xi }    \nonumber \\
  & \times \left\{  \epsilon c_1 \left(\lambda_1^2 -3 \lambda_2 \lambda_1 + 2 \lambda_2^2\right)^2  e^{\lambda_1 \xi }+2 k^2 \lambda_1^2 \left(\lambda_1^2 +3 \lambda_2 \lambda_1 +2 \lambda_2^2 \right)^2 \right\}.   
 \end{align}
 One can easily check that series expansion of generating functions $U(\xi, \epsilon), V(\xi, \epsilon)$ in $\epsilon$ about $0$ provides the correction terms as coefficients of several powers of  $\epsilon$. Consequently, the exact solution of  (\ref{KdVe3}) can be obtained by setting $\epsilon=1$ in  (\ref{mgfun1})-(\ref{mgfun3}) in the form
 \begin{align}\label{mgfun4}
 U(\xi )=&4 k^4 \left(\lambda _1-2 \lambda _2\right) \left(\lambda _1+\lambda _2\right)^2 \left(\lambda _1+2 \lambda _2\right)^2 \left[ 64 c_1  a^2 k^8  \lambda_1^4 \lambda_2^4 \left(\lambda_1 -2  \lambda _2\right) \right. \nonumber  \\
 & \left. \times   \left(\lambda _1+\lambda _2\right)^2  \left(\lambda _1+2 \lambda _2\right)^4  e^{\lambda _1 \xi } - 16   a b k^2 c_2^2 \lambda_2^2 \left(\lambda _1+2 \lambda _2\right)  \left\{  c_1 k^2 \lambda_1^2  \right. \right.  \nonumber \\ 
 & \times  \left. \left. \left(\lambda_1^2-4 \lambda_2^2\right)^2  \left(\lambda_1^2+\lambda _2^2\right)   e^{\left(\lambda _1+2 \lambda _2\right) \xi } + c_1^2 \lambda _2^2 \left(\lambda _1^2-3 \lambda _1 \lambda _2  + 2   \lambda _2^2\right)^2  \right. \right.  \nonumber \\ 
 & \times  \left. \left.  e^{2 \left(\lambda_1+\lambda _2\right) \xi } +4 k^4 \lambda_1^4 \lambda_2^2 \left(\lambda _1^2+3 \lambda _2 \lambda _1+2 \lambda _2^2\right)^2  e^{2 \lambda _2 \xi }\right\}+  b^2  c_1 c_2^4 \lambda_1^4  \right. \nonumber \\
&\times \left. \left(\lambda _1-2 \lambda _2\right) \left(\lambda _1-\lambda _2\right)^2  e^{\left(\lambda _1+4 \lambda_2\right) \xi }\right]\text{/} Q(\xi)^2,  
\end{align}
\begin{align}\label{mgfun5}
V(\xi)=& 8 a c_2 k^4  \lambda_2^2 \left(\lambda _1-2 \lambda _2\right) \left(\lambda _1+\lambda _2\right) \left(\lambda _1+2 \lambda _2\right)^2   e^{\lambda _2 \xi } \left[   c_1 \left(\lambda _1^2-3 \lambda _2 \lambda _1 \right. \right. \nonumber \\
& \left. \left. +2 \lambda_2^2\right) e^{\lambda _1 \xi } +2 k^2 \lambda _1^2  \left(\lambda _1^2+3 \lambda _2 \lambda _1+2 \lambda _2^2\right) \right] \text{/} Q(\xi),
\end{align}
   where
\begin{align}\label{mgfun6}
  Q(\xi )=& 8 a k^4 \lambda_2^2 \left(\lambda _1-2 \lambda _2\right)  \left(\lambda _1+\lambda _2\right)^2 \left(\lambda_1+2 \lambda_2\right)^3  \left(  c_1   e^{\lambda_1 \xi }+2 k^2 \lambda_1^2 \right)-b  c_2^2 e^{2 \lambda _2 \xi }    \nonumber \\
  & \times \left\{   c_1 \left(\lambda_1^2 -3 \lambda_2 \lambda_1 + 2 \lambda_2^2\right)^2  e^{\lambda_1 \xi }+2 k^2 \lambda_1^2 \left(\lambda_1^2 +3 \lambda_2 \lambda_1 +2 \lambda_2^2 \right)^2 \right\}.      
 \end{align} 
 Equations  (\ref{mgfun4})-(\ref{mgfun6}) with (\ref{KdVe2}) and (\ref{KdVe4}) constitutes the solution of (\ref{KDVe1}). These solutions have been verified through symbolic computation.
\section{Boundedness of solution}\label{sec5}
 Bounded solutions have exigent bags of contribution in recounting  perspective of  physical systems  modelled by differential equations. This section is devoted to deriving the boundedness conditions of solution (\ref{mgfun4})-(\ref{mgfun6}) with (\ref{KdVe2}) and (\ref{KdVe4})  of (\ref{KDVe1}). To attain the goal we propose the following theorems:
\begin{thm}\label{th1} \ The solution (\ref{mgfun4})-(\ref{mgfun6}) of  equation (\ref{KdVe3}) is bounded in $\mathbb{R}$ if the parameters $a, \ b, \ k,\ \lambda _1,\  \lambda _2$ involved in the equation and arbitrary constants $c_1, \ c_2$ present in the solution satisfy any one of the following conditions. \\
In case of \textbf{ $0< \frac{1}{2} \lambda _1<\lambda _2< \lambda _1$:}
\begin{description}
\item{I.(a)}\ \ $c_1>0,\ c_2>0,\ a>0,\  b>0,\  k<0,$
\item{I.(b)}\ \ $c_1>0,\ c_2>0,\ a>0,\  b>0,\  k>0,$
\item{I.(c)}\ \ $c_1>0,\ c_2<0,\ a>0,\  b>0,\  k<0,$
\item{I.(d)}\ \ $c_1>0,\ c_2<0,\ a>0,\  b>0,\  k>0.$
\end{description}
 In case of \textbf{ $0<\lambda _1< \lambda _2$ :}
\begin{description}
\item{II.(a)}\ \ $c_1>0,\ c_2>0,\ a>0,\  b>0,\  k<0,$
\item{II.(b)}\ \ $c_1>0,\ c_2>0,\ a>0,\  b>0,\  k>0,$
\item{II.(c)}\ \ $c_1>0,\ c_2<0,\ a>0,\  b>0,\  k<0,$
\item{II.(d)}\ \ $c_1>0,\ c_2<0,\ a>0,\  b>0,\  k>0.$
\end{description}
 In case of \textbf{ $0<\lambda _2<\frac{1}{2}\lambda _1$ :}
\begin{description}
\item{III.(a)}\ \ $c_1>0,\ c_2>0,\ a>0,\  b<0,\  k<0,$
\item{III.(b)}\ \ $c_1>0,\ c_2>0,\ a>0,\  b<0,\  k>0,$
\item{III.(c)}\ \ $c_1>0,\ c_2<0,\ a>0,\  b<0,\  k<0,$
\item{III.(d)}\ \ $c_1>0,\ c_2<0,\ a>0,\  b<0,\  k>0.$
\end{description}
\end{thm}
\begin{proof}
 Solution (\ref{mgfun4})-(\ref{mgfun6}) have a common factor in denominators given by $Q(\xi).$
Boundedness of this solution demands that $Q(\xi)$ never vanishes (ie., $Q(\xi)>0 $ \text{or} $Q(\xi)<0 $). For $c_1>0$, we have
$ \left( c_1   e^{\lambda_1 \xi }+2 k^2 \lambda_1^2\right) >0 $  and
 $$\left\{  c_1 \left(\lambda_1^2 -3 \lambda_2 \lambda_1 +2 \lambda_2^2\right)^2   e^{\lambda_1 \xi } +2 k^2 \lambda_1^2 \left(\lambda_1^2 +3 \lambda_2 \lambda_1 +2 \lambda_2^2 \right)^2 \right\}>0.$$
Above presented conditions with $a>0$ lead us to the conditions for non-vanishing  $Q(\xi)$ as: 
\begin{description} 
\item{Case  $Q(\xi)< 0$:} $ \  \left\{ \left(\lambda _1-2 \lambda _2\right)<0,\  b>0 \right\} \Rightarrow  \ \left\{  \frac{1}{2} \lambda_1<\lambda_2<\lambda_1,\ b>0 \right\}$ \ or \ $ \left\{ 0< \lambda_1<\lambda_2,\ b>0 \right\}. $\\
\item{Case  $Q(\xi)> 0$:} $ \  \left\{ \left(\lambda _1-2 \lambda _2\right)>0, \ b<0 \right\} \Rightarrow  \  \left\{   0<\lambda_2< \frac{1}{2} \lambda_1,\ b<0 \right\}. $
\end{description}
Case $Q(\xi)< 0$ provide the conditions I.(a)$-$I.(d) and II.(a)$-$II.(d) whereas  Case $Q(\xi)> 0$ give conditions III.(a)$-$III.(d) of the theorem. That completes the proof of the theorem.
\end{proof}
To check the correctness of the conditions presented  in the above theorem, few particular values of parameters satisfying above conditions have been presented in the table \ref{tab1}. These values have been utilized in figure \ref{fig1} to plot the solution (\ref{mgfun4})-(\ref{mgfun6}). It is clear from the figures that the solution $U(\xi)$ (given in (\ref{mgfun4})) have three-pick soliton like shape for the conditions I.(a)$-$I.(d) and II.(a)$-$II.(d) respectively and for remaining cases, it bears two-pick soliton form.  However $V(\xi)$ (given in (\ref{mgfun5})) always restrain two-pick soliton form. Besides that, it is important to note here that  these solutions may bear a lower number of picks for different parameters values satisfying above conditions. 
\begin{table}
\begin{center}
\caption{ Values of parameters present in $U(\xi),\ V(\xi)$ satisfying conditions presented in Theorem \ref{th1}, used  in Figure \ref{fig1}.}\label{tab1}
\resizebox{\linewidth}{!}{ 
\begin{tabular}{l l l l l l l l c}
\hline
Case &$\lambda_1$ & $\lambda_2$ & $a$ &$b$ &$k$&$c_1$&$c_2$&Figure\\[7pt] \hline \hline
&$ 0.4 $& $ 0.399 $& $1.6 $ & $0.1 $ &$ -0.25 $&$ 1.9 $&$0.5 $&I.(a)\\ 
$\frac{1}{2} \lambda_1<\lambda_2<\lambda_1$&$ 1.43 $& $ 1.42 $& $1 $ & $3.1 $ &$ 0.5 $&$ 4.9 $&$0.6 $&I.(b)\\ 
$a>0, \ b>0$&$ 0.4 $& $ 0.399 $& $1.6 $ & $0.1 $ &$ -0.25 $&$ 1.9 $&$-0.5 $&I.(c)\\ 
&$ 1.43 $& $ 1.42 $& $1 $ & $3.1 $ &$ 0.5 $&$ 4.9 $&$-0.6 $&I.(d)\\ \hline \hline 

&$ 0.38888 $& $ 0.3889 $& $2.1 $ & $0.001 $ &$ -0.02 $&$ 8.9 $&$9.4 $&II.(a)\\ 
$0< \lambda_1<\lambda_2$&$ 1.08 $& $ 1.0809 $& $0.005 $ & $0.0055 $ &$ 0.2 $&$ 0.039 $&$0.003 $&II.(b)\\ 
$a>0, \ b>0$&$ 0.38888 $& $ 0.3889 $& $2.1 $ & $0.001 $ &$ -0.02 $&$ 8.9 $&$-9.4 $&II.(c)\\ 
&$ 1.08 $& $ 1.0809 $& $0.005 $ & $0.0055 $ &$ 0.2 $&$ 0.039 $&$-0.003 $&II.(d)\\ \hline \hline 

&$ 1.11 $& $ 0.42 $& $1.2 $ & $-3.1 $ &$ -1.1 $&$ 2.3 $&$0.3 $&III.(a)\\ 
$0<\lambda_2< \frac{1}{2} \lambda_1$&$ 2 $& $ 0.99 $& $1.2 $ & $-3.1 $ &$ 1.1 $&$ 1.3 $&$0.3 $&III.(b)\\ 
$a>0, \ b<0$&$ 1.11 $& $ 0.42 $& $1.2 $ & $-3.1 $ &$ -1.1 $&$ 2.3 $&$-0.3 $&III.(c)\\ 
&$ 2 $& $ 0.99 $& $1.2 $ & $-2.1 $ &$ 1.1 $&$ 5.3 $&$-0.6 $&III.(d)\\ \hline \hline
\end{tabular}
}
\end{center}
\end{table}
\begin{figure*}
\centering
\includegraphics[width=0.24\textwidth]{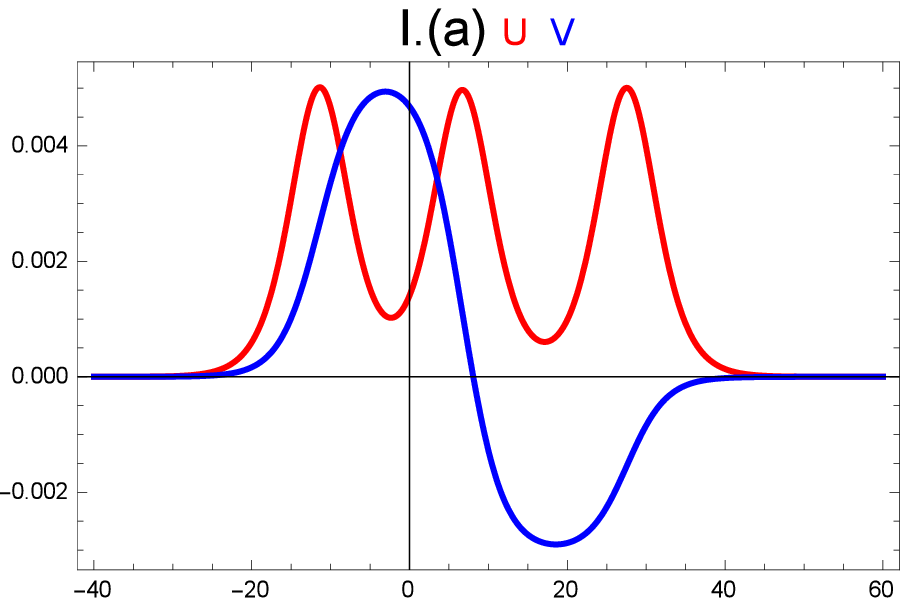}
\includegraphics[width=0.24\textwidth]{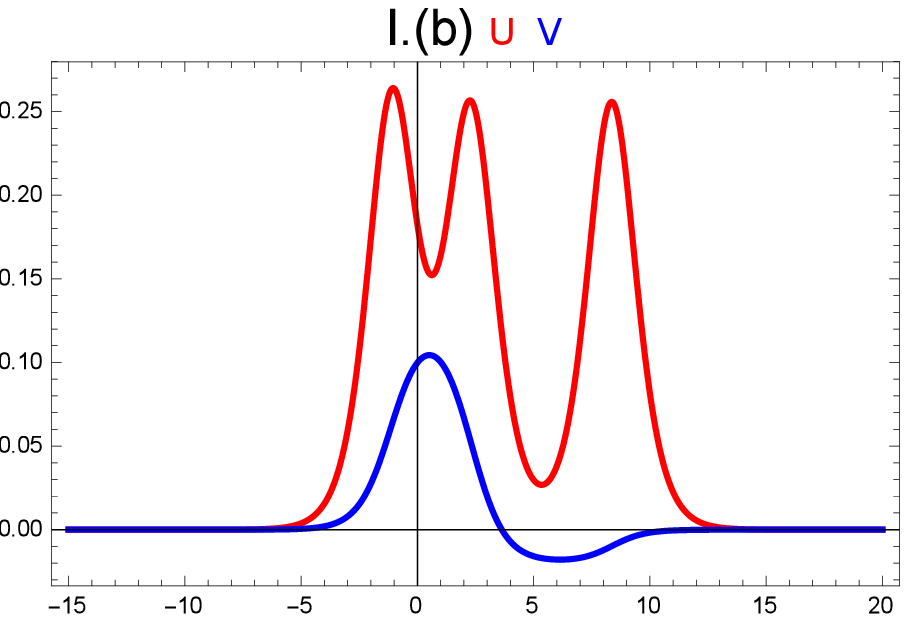}
\includegraphics[width=0.24\textwidth]{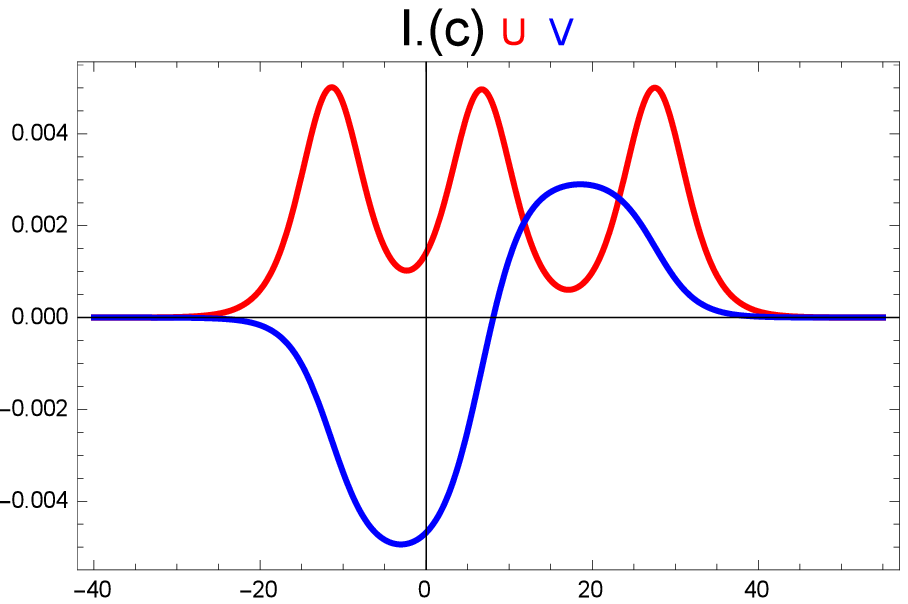}
\includegraphics[width=0.24\textwidth]{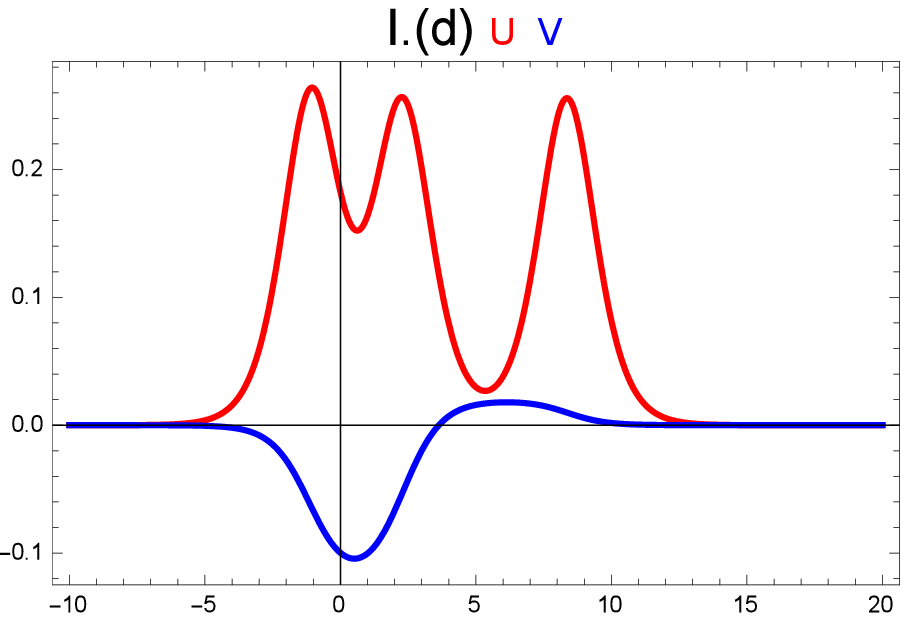}
 \vskip 0.1in 
\includegraphics[width=0.24\textwidth]{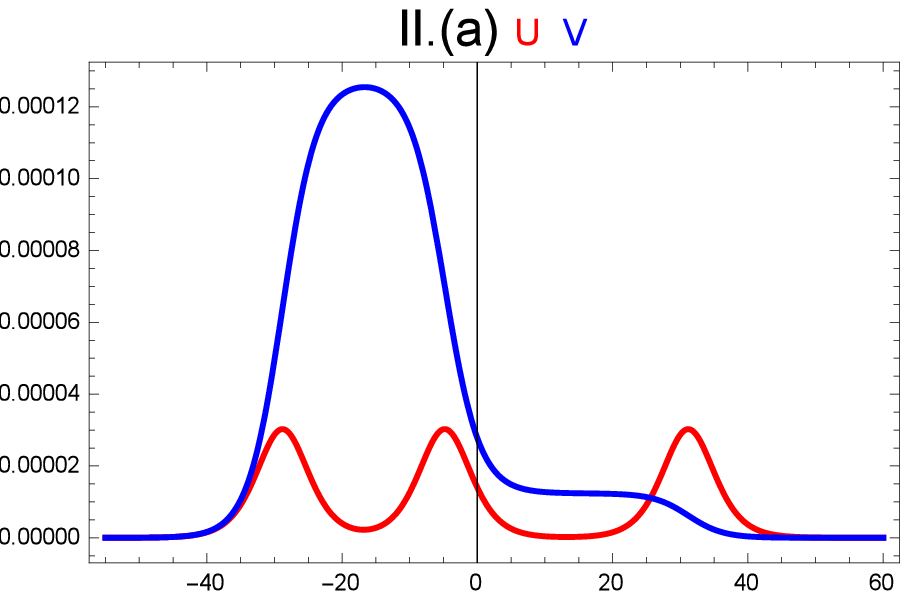}
\includegraphics[width=0.24\textwidth]{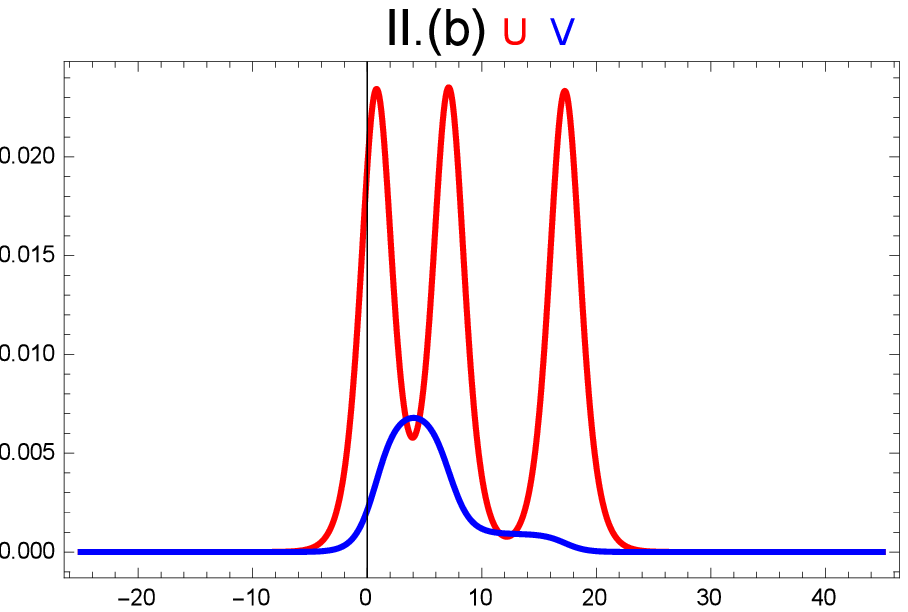}
\includegraphics[width=0.24\textwidth]{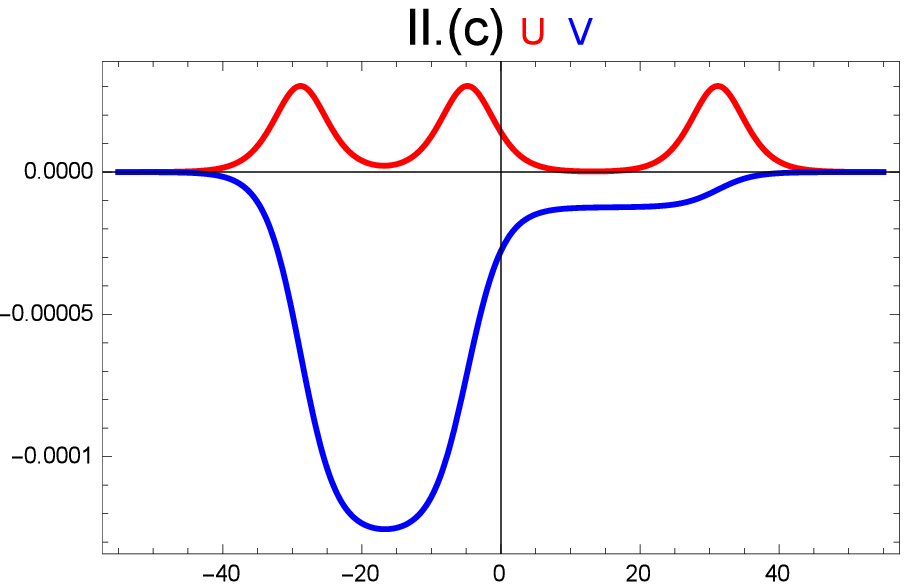}
\includegraphics[width=0.24\textwidth]{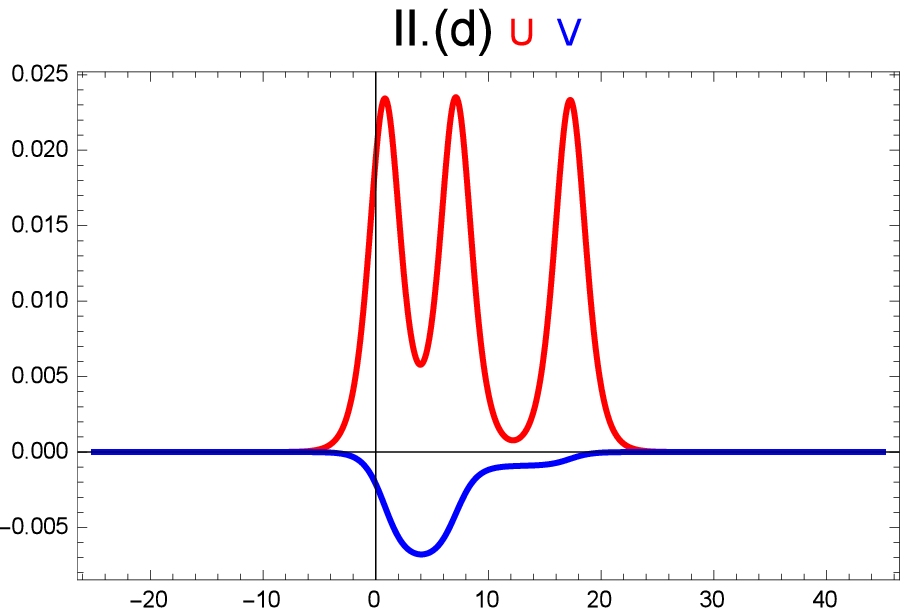}
 \vskip 0.1in
\includegraphics[width=0.24\textwidth]{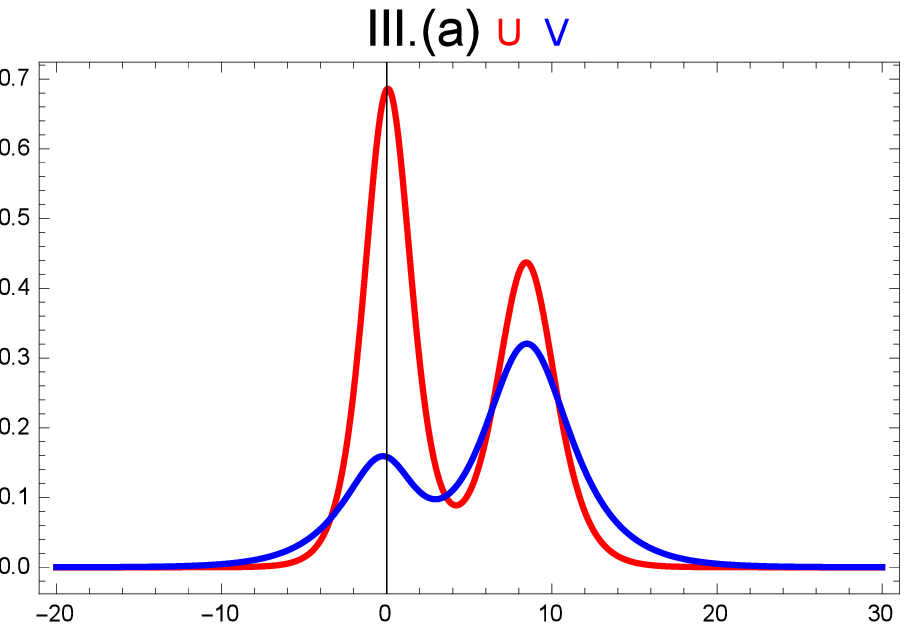}
\includegraphics[width=0.24\textwidth]{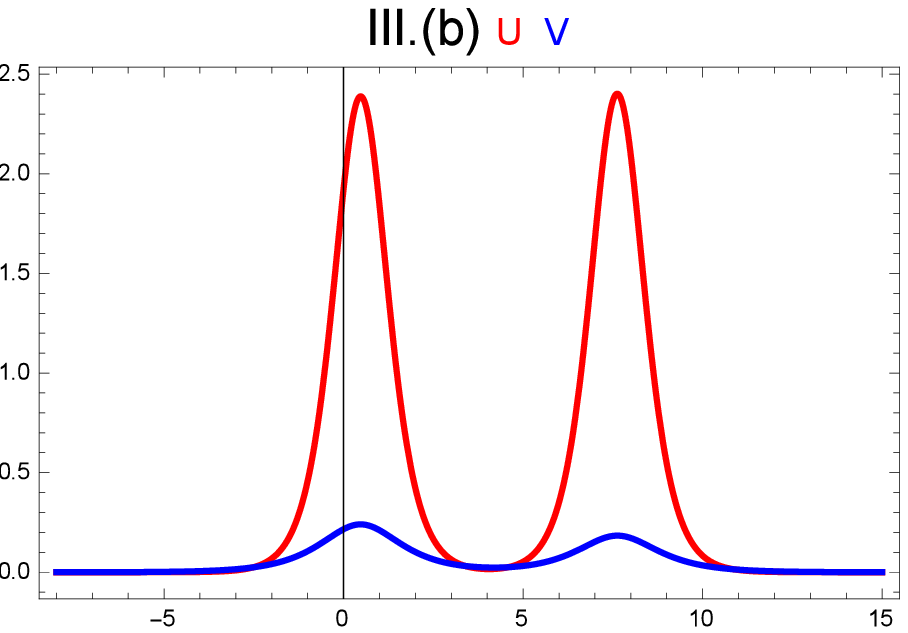}
\includegraphics[width=0.24\textwidth]{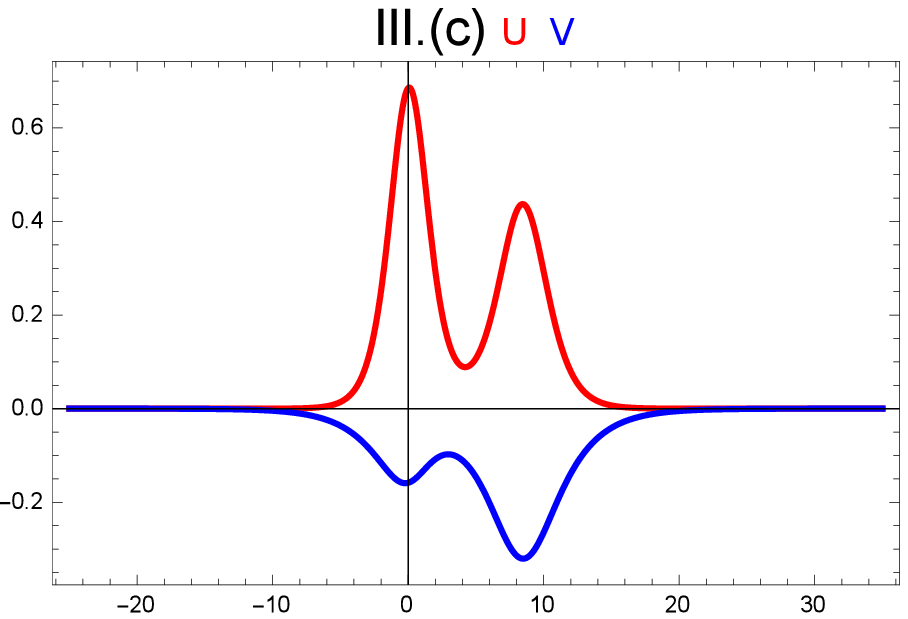}
\includegraphics[width=0.24\textwidth]{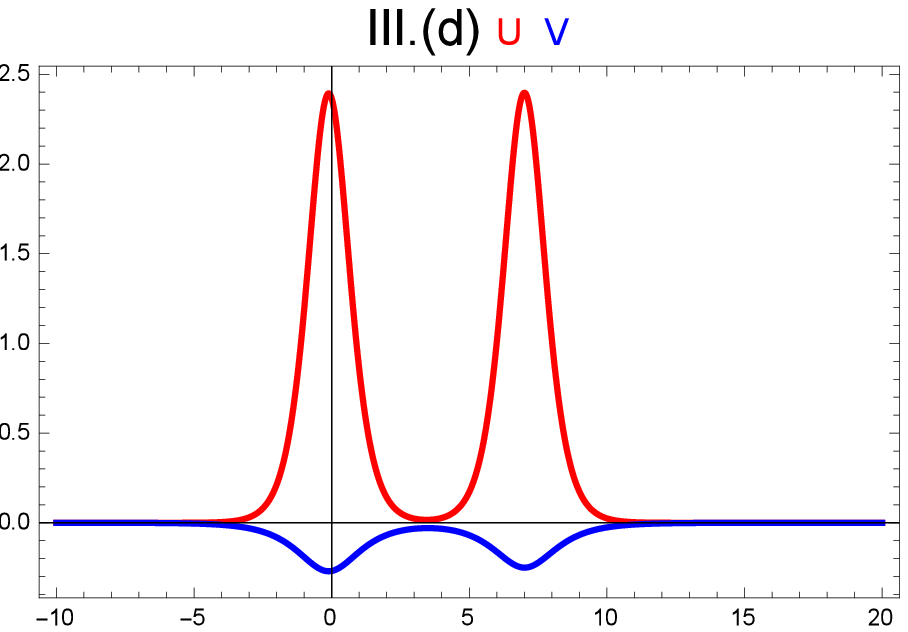}
\vskip 0.05in 
\caption{Plots of solution (\ref{mgfun4})-(\ref{mgfun5}) for the conditions presented in Theorem \ref{th1}, using values of parameters given in
   Table \ref{tab1}.}\label{fig1}       
\end{figure*}
\begin{thm}\label{th2} \ The solution (\ref{mgfun4})-(\ref{mgfun6}) with (\ref{KdVe2}) and (\ref{KdVe4})  of  equation (\ref{KDVe1}) is bounded in $\mathbb{R} \times \mathbb{R}$ if the parameters $a, \ b, \ c,\ k$ involved in the equation and arbitrary constants $c_1, \ c_2$ present in the solution satisfy any one of the following conditions. \\
In case of \textbf{ $\frac{1}{2}\sqrt{-\frac{c}{a k^3}}<\sqrt{-\frac{c}{ k^3}}< \sqrt{-\frac{c}{a k^3}}$:}
\begin{description}
\item{I.(a)}\ \ $c_1>0,\ c_2>0,\ a>0,\  b>0,\  k<0,\  c>0,$
\item{I.(b)}\ \ $c_1>0,\ c_2>0,\ a>0,\  b>0,\  k>0,\ c<0,$
\item{I.(c)}\ \ $c_1>0,\ c_2<0,\ a>0,\  b>0,\  k<0,\  c>0,$
\item{I.(d)}\ \ $c_1>0,\ c_2<0,\ a>0,\  b>0,\  k>0,\  c<0.$
\end{description}
 In case of \textbf{ $0<\sqrt{-\frac{c}{a k^3}}< \sqrt{-\frac{c}{ k^3}}$ :}
\begin{description}
\item{II.(a)}\ \ $c_1>0,\ c_2>0,\ a>0,\  b>0,\  k<0,\  c>0,$
\item{II.(b)}\ \ $c_1>0,\ c_2>0,\ a>0,\  b>0,\  k>0,\  c<0,$
\item{II.(c)}\ \ $c_1>0,\ c_2<0,\ a>0,\  b>0,\  k<0,\  c>0,$
\item{II.(d)}\ \ $c_1>0,\ c_2<0,\ a>0,\  b>0,\  k>0,\  c<0.$
\end{description}
 In case of \textbf{ $0<\sqrt{-\frac{c}{ k^3}}<\frac{1}{2}\sqrt{-\frac{c}{a k^3}}$ :}
\begin{description}
\item{III.(a)}\ \ $c_1>0,\ c_2>0,\ a>0,\  b<0,\  k<0,\  c>0,$
\item{III.(b)}\ \ $c_1>0,\ c_2>0,\ a>0,\  b<0,\  k>0,\  c<0,$
\item{III.(c)}\ \ $c_1>0,\ c_2<0,\ a>0,\  b<0,\  k<0,\  c>0,$
\item{III.(d)}\ \ $c_1>0,\ c_2<0,\ a>0,\  b<0,\  k>0,\ c<0.$
\end{description}
\end{thm}
\begin{proof}
This theorem can be proved in a straightforward way by using equations (\ref{KdVe2})-(\ref{KdVe4}) and  theorem \ref{th1}.
\end{proof}
To verify the conditions presented  in theorem \ref{th2}, some particular values of parameters satisfying the above conditions have been exhibited in table \ref{tab2}. Utilizing those values, 3D plots of the solution  (\ref{mgfun4})-(\ref{mgfun6}) with (\ref{KdVe2}) and (\ref{KdVe4}) for different cases have displayed in figure \ref{fig2}. Figures reveal that the solution $u(x,t)$ (given by (\ref{mgfun4})-(\ref{mgfun6}) with (\ref{KdVe2}) and (\ref{KdVe4})) always have three-hump soliton like  features for the conditions I.(a)$-$I.(d) and II.(a)$-$II.(d) respectively whereas it bears two-hump soliton shape for remaining cases of theorem \ref{th2}.  However $v(x,t)$ (given by (\ref{mgfun5})-(\ref{mgfun6}) with (\ref{KdVe2}) and (\ref{KdVe4})) always restrain two-pick form. Besides that, it is important to mention that  these solutions may bear a lower number of humps. From the 3D plots it is clear that solution can describe various multi-pick soliton states for diverse  values of the corresponding free parameters.

The existing multi-pulse solutions in the literature \cite{chugunova2007two,groves1998solitary}, are constituted with multiple copies of the one-pulse solutions separated by finitely many oscillations close to the zero equilibrium (has oscillatory decaying tails at infinity). It is worth noting that here obtained  multi-pulse solutions ((\ref{mgfun4})-(\ref{mgfun6})) of equation (\ref{KDVe1})  constitute multiple picks but do not oscillate close to the zero equilibrium. This is due to the fact that auxiliary equations of the linear parts of the considered coupled nonlinear differential equations (\ref{KdVe3}) have distinct real roots $(\lambda_1 \neq \lambda_2)$. That ensures  the nonexistence of any periodic functions in the derived solution. So in view of study \cite{groves1998solitary} and above property lead us to conclude that derived solutions (Eq. (\ref{mgfun4})-(\ref{mgfun6})) are  homoclinic  and always have finite numbers of humps and their tails decay to zero exponentially in a monotonic fashion, as depicted in figures \ref{fig1} and  \ref{fig2}. This fact, differentiate the derived solutions (exact) of this paper with the existing multi-hump stationary wave solutions (numerical) in literature  \cite{chugunova2007two,groves1998solitary}.  Also, it is important to mention here that usually soliton trains arrive straight in shape but in  figure \ref{fig2} they are appeared in bent shape due to the existence of  conformable derivative parameter  $\alpha$ and $\beta$ in the solution. 
\begin{table*}
\caption{ Values of parameters parameters in $u(x,t),\ v(x,t)$ satisfying conditions presented in Theorem \ref{th2}, used  in Figure \ref{fig2}.}\label{tab2}
\begin{center}
\resizebox{\linewidth}{!}{ 
\begin{tabular}{l l l l l l l l l l c}
\hline
Case  & $a$ &  $b$ & $c$ & $k$ &$c_1$ &  $c_2$ &  $\alpha$  &  $\beta$ & $\xi_0$&Figure\\[7pt] \hline \hline
&$ 0.990 $& $ 0.09 $& $1.0 $ & $-1.8 $ &$ 4.6 $&$.58 $&$.68$&  $.69$& $28$ &I.(a)\\
\tiny{$\frac{1}{2}\sqrt{-\frac{c}{a k^3}}<\sqrt{-\frac{c}{ k^3}}< \sqrt{-\frac{c}{a k^3}}$}&$ 0.980 $& $ 0.09 $& $-1.6 $ & $2.1 $ &$ 2.8 $&$.58 $&$.70$&  $.75$& $28$ &I.(b)\\
$a>0, \ b>0$&$ 0.989 $& $ 0.02 $& $1.7 $ & $-1.9 $ &$ 3.8 $&$-.9 $&$.82$&  $.80$& $28$ &I.(c)\\
&$ 0.999 $& $ 0.02 $& $-1.7 $ & $1.9 $ &$ 3.0 $&$-.8 $&$.75$&  $.85$& $28$ &I.(d)\\\hline \hline 

&$ 1.001 $& $ 0.001 $& $1.0 $ & $-1.49 $ &$ 6.9 $&$3.4 $&$.79$&  $.69$& $12$ &II.(a)\\
$0<\sqrt{-\frac{c}{a k^3}}<\sqrt{-\frac{c}{ k^3}}$&$ 1.002 $& $ 0.002 $& $-1.5 $ & $2.9 $ &$ .0069 $&$.0035 $&$.80$&  $.92$& $12$ &II.(b)\\[7pt] 
$a>0, \ b>0$&$ 1.001 $& $ 0.0007 $& $2.0 $ & $-2.5 $ &$ .0069 $&$-.0035 $&$.88$&  $.82$& $12$ &II.(c)\\
&$ 1.003 $& $ 0.001 $& $-1.5 $ & $2.9 $ &$ .0068 $&$-.0034 $&$.80$&  $.90$& $12$ &II.(d)\\\hline \hline 

&$ 0.20 $& $- 3.0 $& $2.1 $ & $-2.0 $ &$ 3.1 $&$.3 $&$.70$&  $.80$& $4$ &III.(a)\\
$0<\sqrt{-\frac{c}{ k^3}}<\frac{1}{2}\sqrt{-\frac{c}{a k^3}}$&$ 0.22 $& $ -2.4 $& $-2.0 $ & $1.4 $ &$2.9 $&$.5 $&$.80$&  $.90$& $4$ &III.(b)\\
$a>0, \ b<0$&$ 0.23 $& $ -2.6 $& $2.2 $ & $-1.3 $ &$ 3.3 $&$-.4 $&$.90$&  $.98$& $4$ &III.(c)\\
&$ 0.24 $& $ -2.5 $& $-2.0 $ & $1.3 $ &$ 2.9 $&$-.5 $&$.80$&  $.85$& $4$ &III.(d)\\ \hline \hline
\end{tabular}
}
\end{center}
\end{table*}
\begin{figure*}
\centering
\includegraphics[width=0.24\textwidth]{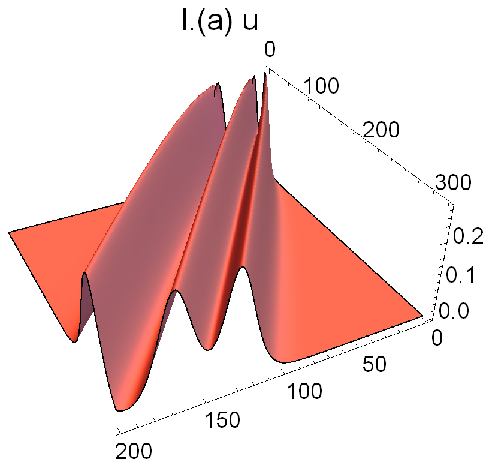}
\includegraphics[width=0.24\textwidth]{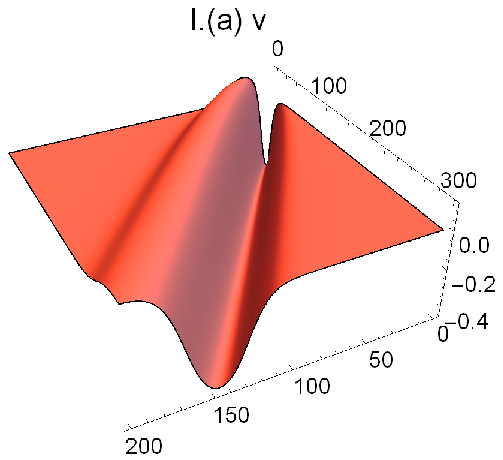}
\includegraphics[width=0.24\textwidth]{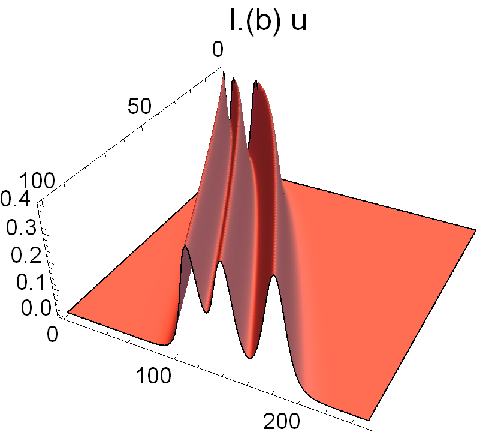}
\includegraphics[width=0.24\textwidth]{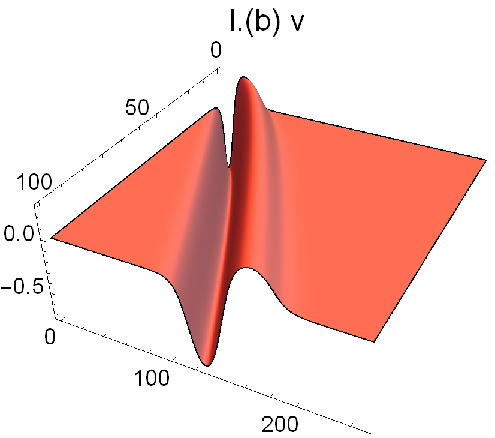}
\vskip 0.1in 
\includegraphics[width=0.24\textwidth]{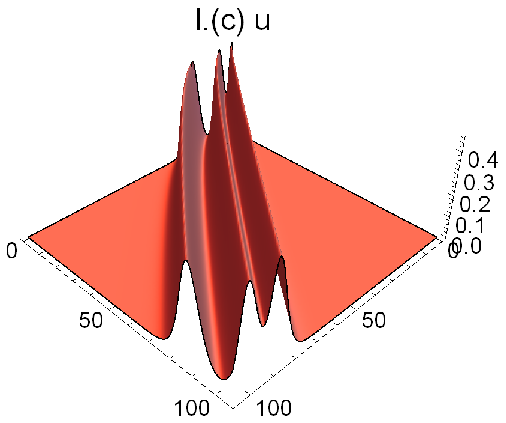}
\includegraphics[width=0.24\textwidth]{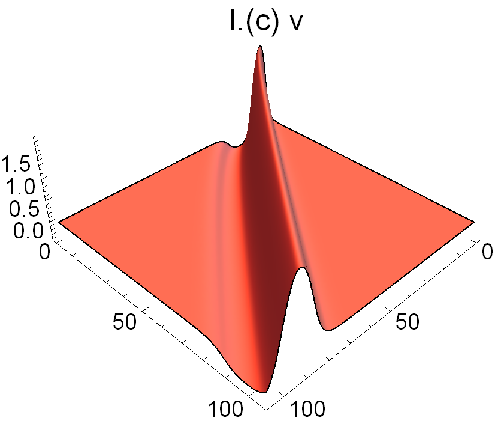}
\includegraphics[width=0.24\textwidth]{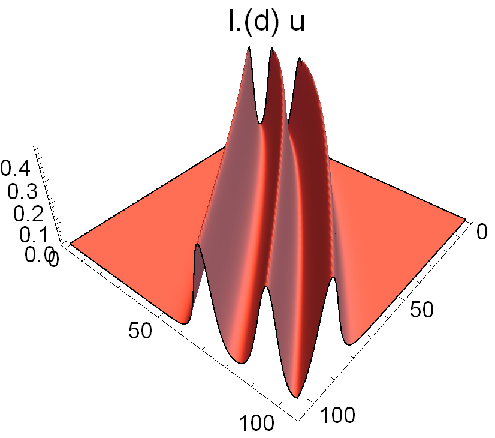}
\includegraphics[width=0.24\textwidth]{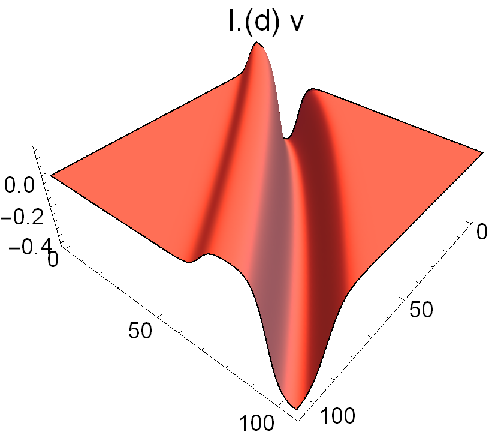}
\vskip 0.1in  
 \includegraphics[width=0.24\textwidth]{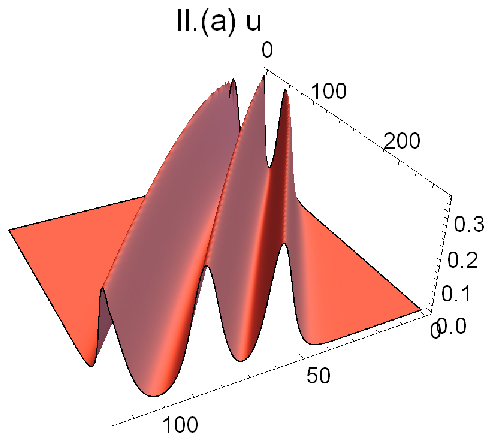}
\includegraphics[width=0.24\textwidth]{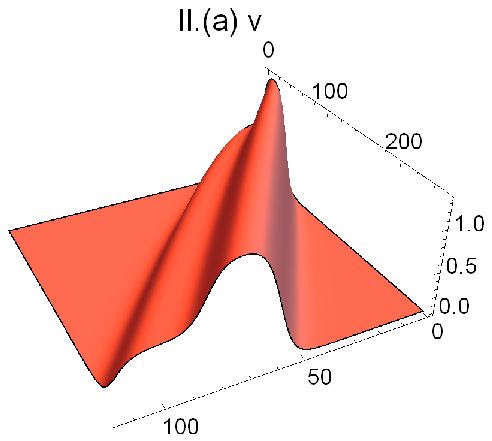}
\includegraphics[width=0.24\textwidth]{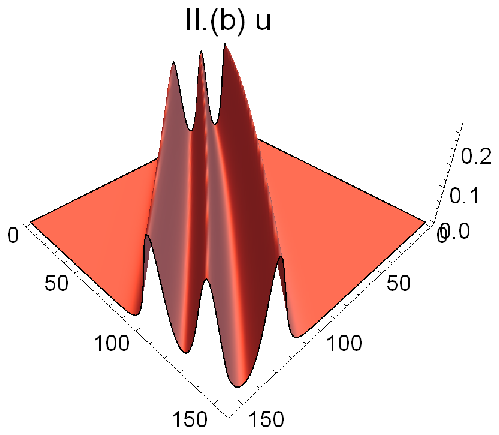}
\includegraphics[width=0.24\textwidth]{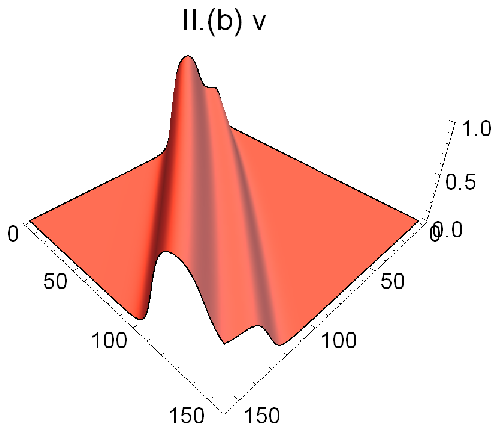}
 \vskip 0.1in 
\includegraphics[width=0.24\textwidth]{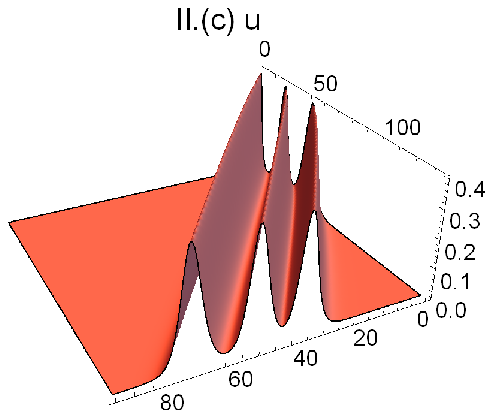}
\includegraphics[width=0.24\textwidth]{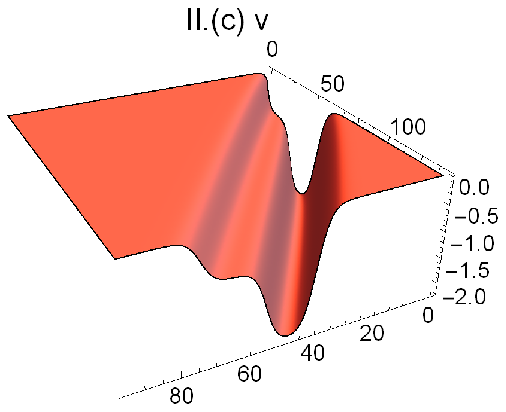}
\includegraphics[width=0.24\textwidth]{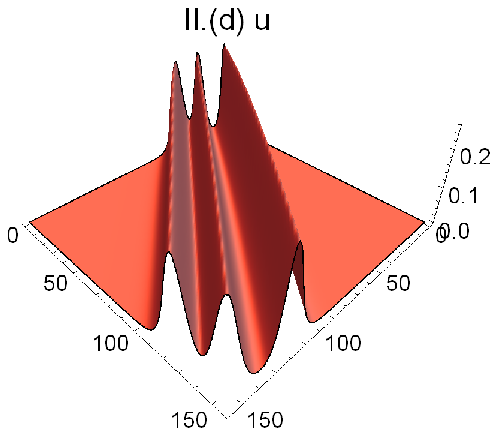}
\includegraphics[width=0.24\textwidth]{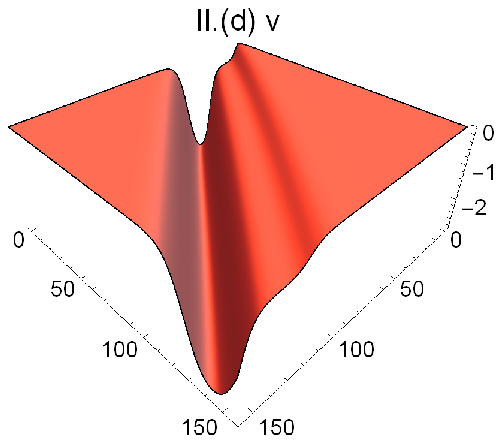}
 \vskip 0.1in  
  \includegraphics[width=0.24\textwidth]{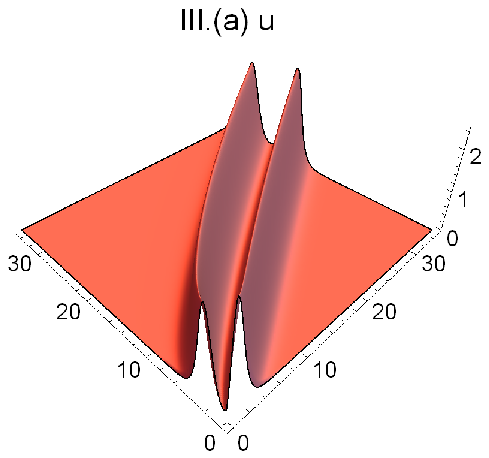}
\includegraphics[width=0.24\textwidth]{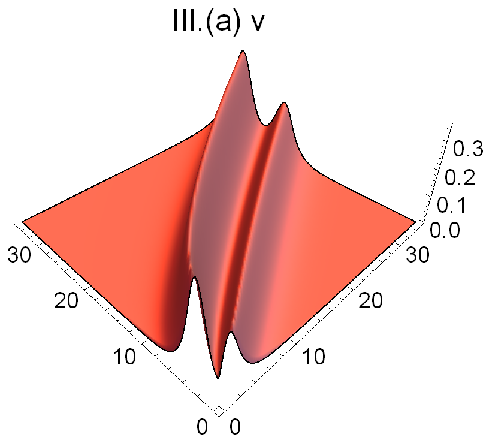}
\includegraphics[width=0.24\textwidth]{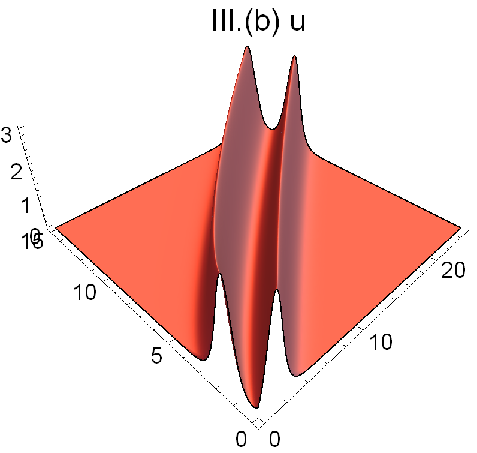}
\includegraphics[width=0.24\textwidth]{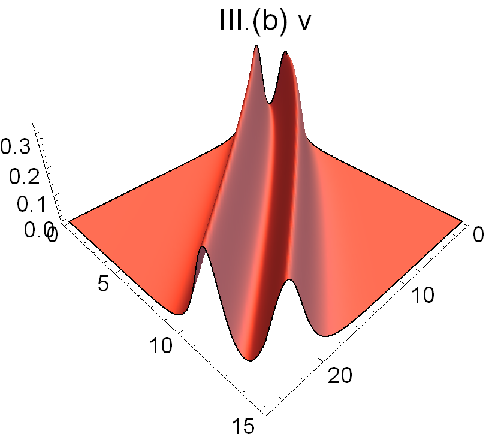}
 \vskip 0.1in 
\includegraphics[width=0.24\textwidth]{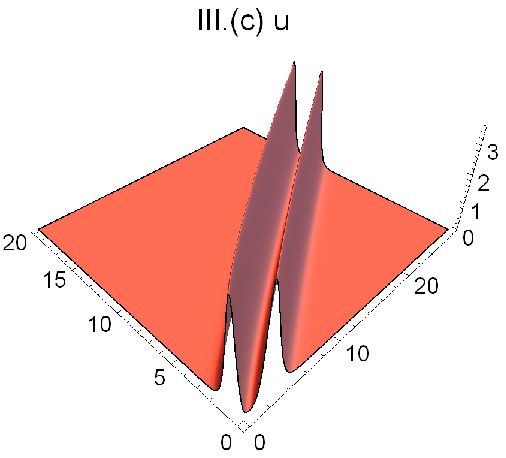}
\includegraphics[width=0.24\textwidth]{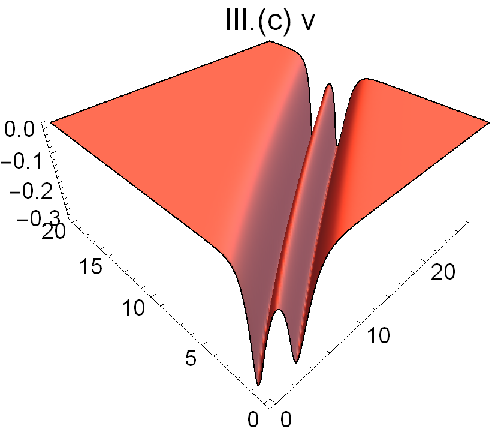}
\includegraphics[width=0.24\textwidth]{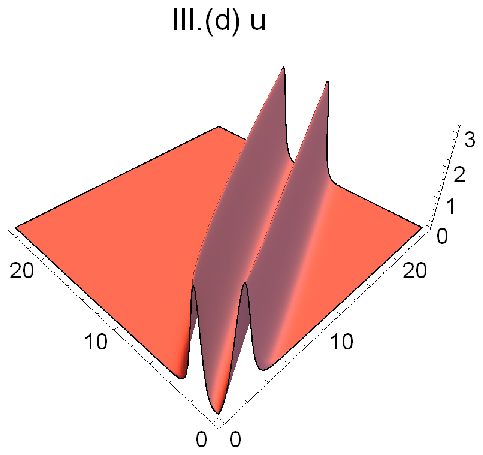}
\includegraphics[width=0.24\textwidth]{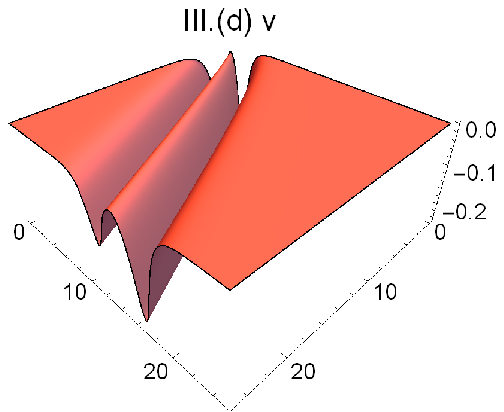}
\vskip 0.05in
\caption{Plots of solution $u(x,t),\ v(x,t)$ for the conditions presented in Theorem \ref{th2}, using values of the parameters submitted  in
   Table \ref{tab2}.}\label{fig2}       
\end{figure*}
\section{ Conclusions}\label{sec6}
In this study, the RCAM was applied to a  coupled Korteweg-de Vries equations with conformable  derivative, which is the generalisation of the mathematical model of waves of shallow water surface equations. A novel exact solution in terms of exponential function has been derived.  Additionally, a few theorems have been presented to predicts the boundedness of the obtained solution. All families of boundedness conditions on the parameters present in the equation and integration constants present in the solutions were tested by plot, affirming their correctness. It was found that there exist triple-pulse, double-pulse, and single-pulse solitons for the considered equation.   The author believes that it is a new finding for the coupled Korteweg-de Vries equations with conformable derivative.
 In addition to that here for the first time, a modification of RCAM  is proposed to deal with couple nonlinear differential equations whose auxiliary equations of linear parts  have distinct roots ($\lambda_1 \neq \lambda_2$), and successfully derived exact multi-hump solutions of coupled KdV Eq. (\ref{KDVe1}).  That extends its applicability to deal with complex nonlinear equations and produce new featured solutions like multi-hump solitons.

The reported results of the present work contribute a new  stationary wave solution  to the class of coupled nonlinear systems  with conformable derivative that may have a composition of the different multi-hump soliton-type features depending on the several restrictions in the parameters present on the equations.  Here we have successfully derived three-hump, two-hump, and one-hump soliton solutions of Eq. (\ref{KDVe1}),  but unable to obtain the solitons having four and higher-order hump due to the unavailability of more wave parameters. It is straightforward to conclude that the rapidly convergent approximation method is a powerful tool to solve various nonlinear models involving conformable derivative.

 The modification of RCAM to reveal solitary wave and multi-soliton solutions of constant and variable coefficient differential equations is significant in the studied field and might have an important impact on future research. So in the near future, the author wants to modify the scheme for exhaling solitary wave and multi-soliton solutions of constant and variable coefficient differential equations.

\section*{Acknowledgements}
The author thanks the editor and reviewers for their comments and suggestions to improve the paper in the revised form.
\section*{Conflict of interest}
The authors declare that they have no conflict of interest.
\section*{References}
\bibliographystyle{iopart-num}      
\bibliography{CFKdV}   
\end{document}